\documentclass[11pt]{article}

\usepackage{amsmath,amsfonts,amsthm,amssymb,bm} % color

\usepackage{fullpage}
\newtheorem{theorem}{Theorem}[section]
\newtheorem{lemma}[theorem]{Lemma}
\newtheorem{definition}{Definition}[section]

\usepackage{pstricks,pst-node}

\usepackage{mathtools}

\newcommand{\NN}{\mathbb{N}}
\newcommand{\RR}{\mathbb{R}}
\newcommand{\eps}{\epsilon}

\begin{document}
\title{A Deterministic Protocol for Sequential Asymptotic Learning\thanks{
Part of this work was done while Yu Cheng was a student at the University of Southern California.
Yu Cheng was supported in part by Shang-Hua Teng's Simons Investigator Award.
Omer Tamuz was supported in part by Grant \#419427 from the Simons Foundation.}} 
\author{
  Yu Cheng \\ Duke University \and
  Wade Hann-Caruthers \\ California Institute of Technology \and
  Omer Tamuz \\ California Institute of Technology
}
\date{}

\maketitle

%!TEX root = paper.tex

\begin{abstract}
In the classic herding model, agents receive private signals about an underlying binary state of nature, and act sequentially to choose one of two possible actions,  after observing the actions of their predecessors.

We investigate what types of behaviors lead to asymptotic learning,
where agents will eventually converge to the right action in probability.
It is known that for rational agents and bounded signals, there will not be asymptotic learning. 

Does it help if the agents can be cooperative rather than act selfishly? This is simple  to achieve if the agents are allowed to use randomized protocols. In this paper, we provide the first \emph{deterministic} protocol under which asymptotic learning occurs. In addition, our protocol has the advantage of being much simpler than previous protocols.
\end{abstract}

%!TEX root = paper.tex

\section{Introduction}

When making decisions, we often have some amount of information which we have learned on our own, but we also look to what choices other people have made in the past.
In sequential learning models, there is a sequence of agents $\NN = \{1,2,\ldots\}$ who are interested in learning an unknown state of the world $\theta\in\{0,1\}$. Each agent $i$ receives a signal $s_i$ that depends on the state, takes an action $a_i \in \{0,1\}$, and receives a utility of $1$ if the action matches the state, and utility $0$ otherwise. Each agent also observes the actions of her predecessors before choosing her own. 

The classical result of~\cite{banerjee1992simple,bikhchandani1992theory} is that from some point on, rational agents will disregard their own private signal and emulate the actions of their predecessors, resulting in a {\em herd} in which almost all the information is lost. In particular, {\em asymptotic learning} does not occur: the probability that the $i$\textsuperscript{th} agent chooses the correct action does not tend to one, as it would if the agents could observe the signals (rather than the actions) of their predecessors. 

A natural question is whether non-selfish agents can coordinate on a protocol in which they do not always choose the action that is optimal given their information, but where asymptotic learning does occur. Achieving this with probabilistic protocols is straightforward~\cite{acemoglu2011bayesian,drakopoulos2013learning}: for example~\cite{acemoglu2011bayesian}, one could have the $i$\textsuperscript{th} agent choose an action that reveals her private signal with probability $1/i$, and choose her optimal action with probability $1-1/i$.

In this paper, we show that asymptotic learning can also be achieved using a deterministic protocol.

\begin{theorem}
\label{thm:main}
There is a deterministic protocol that achieves asymptotic learning.
\end{theorem}

At the heart of the problem is a trade-off between optimization and communication. In order for asymptotic learning to occur, agents must rely more and more on the information provided by their predecessors' actions, since their own private signals have a uniformly bounded amount of information. However, if agents rely too much on previous observations, they will pass too little information on through their own actions. The question of whether there are any deterministic protocols under which asymptotic learning occurs is essentially the following: are there any deterministic protocols which can strike the right balance between agents taking optimal actions and agents using their actions to communicate their private information to their successors?

% These facts help illustrate what we refer to as the communication-optimization tradeoff. In order for learning to occur, agents must rely more and more on the information provided by their predecessors' actions. However, if an agent rely too much on this information, she will pass too little information on through her own action, and thus will not contribute enough information to the aggregate information pool for her action to help later agents.

% Thus, the question of whether there are any protocol profiles under which learning occurs may be intuitively understood as: are there any protocol profiles which can strike the right balance between agents taking optimal actions and agents using their actions to communicate information provided by their private signals to their successors?

Our work raises some potentially interesting questions: in what reasonable economic settings do selfish agents naturally coordinate and achieve asymptotic learning in equilibrium? And how could one incentivize them to do so in other settings? 

% We consider the main contribution of our result to be the following. It is always possible for agents to act in such a way that agents' likelihood of taking the correct action converges to 1, even when they can communicate only through their actions, but this requires cooperation. What prevents this result from obtaining in the herding literature with bounded signals is that agents have no incentive to cooperate; agents do not value later agents' likelihood of taking the correct action. This result also prompts the question: are there economically reasonable variants on the classical sequential learning model where agents do assign some value to later agents taking the correct action, so that learning occurs in equilibrium?

\subsection{Related Work}
The classical sequential learning papers~\cite{banerjee1992simple,bikhchandani1992theory} have been followed by a large literature in Economics (see, e.g.,~\cite{smith2000pathological,acemoglu2011bayesian,eyster2010naive,hann2017speed}, as well as Chamley's book~\cite{chamley2004rational}) and  engineering (see, e.g.,~\cite{lobel2009rate,drakopoulos2013learning}). Smith and S{\o}rensen~\cite{smith2000pathological} showed that when private signals are unbounded (i.e., there is no limit to how strong an indication a private signal can give) then agents do learn the state in the sequential setting; this observation has motivated us---and others---to understand under what other conditions learning can be achieved. As far as we know,~\cite{drakopoulos2013learning} is the only other paper that studies related algorithmic questions. In particular, in~\cite{drakopoulos2013learning} it is shown how agents with bounded memory who observe only some of their predecessors can achieve asymptotic learning using a randomized protocol.

% - High level description of problem
% - Literature and background where problem comes from
%   + Especially mention that with bounded signals, there are basically only negative results
% - Motivation for the problem
% - All the extant bounded herding literature discusses only Bayesians who see some subset of their predecessors. BROADER!
% - We're really investigating the limits of the tradeoff between doing what's best given my information and passing on as much information as possible to my successor
% - Discuss "cooperative" nature of our result
% - Advertise our solution to the problem as both a simple algorithm and as an easy to understand but informative observation about herding models

%!TEX root = paper.tex

\section{Model}
\label{sec:model}
There is a countably infinite set of agents indexed by $i \in \NN$,
  acting sequentially and trying to learn the hidden state of nature $\theta$.
In this paper, we focus on the most basic setting where both the underlying state and the private signals are binary.

The state of nature $\theta$ is a random variable which takes values in $\{0, 1\}$.
%\footnote{We do not assume that each state occurs with equal probability.}.
There are two Bernoulli distributions $D_{0}$ and $D_{1}$ with parameters $q_0$ and $q_1$ respectively, i.e., $D_{0}$ takes the value 1 with probability $q_0$ and similarly for $D_{1}$.
The agents know the distributions $D_{0}$ and $D_{1}$, but not the state of nature $\theta$.
Each agent $i$ receives a private signal $s_i \in \{0, 1\}$.
The signals are i.i.d conditioned on $\theta$, drawn from the distribution $D_{\theta}$.
We are interested in the case where the private signals are bounded, that is, both signals are possible regardless of the state.
We assume without loss of generality that $0 < q_0 < q_1 < 1$.
A simple example is that there is a 60/40 biased coin,
  but it could be either biased to flip more heads or tails.
Agents receive independent coin flips and together want to distinguish which coin it is.

The agents act sequentially with each agent $i$ choosing an action $a_i \in \{0, 1\}$.
Before making her choice, agent $i$ can observe her own signal $s_i$,
  as well as the \emph{actions} $\{a_1, \ldots, a_{i-1}\}$ of the previous agents,
  but crucially \emph{not} the private signals of others.
We say that there is \emph{asymptotic learning} if the actions of the agents converge to the true state $\theta$ in probability.

\begin{definition}[Asymptotic Learning]
\label{def:learning}
Let $p_i = \Pr[a_i = \theta]$ be the probability of agent $i$ choosing the correct action.
There is asymptotic learning if and only if $\liminf_{i \rightarrow \infty} p_i = 1$.
\end{definition}

We want to investigate whether asymptotic learning is possible when the agents must act deterministically.
Formally, the agents can agree on a protocol $\Pi = \{f_i\}_{i \in \NN}$, i.e., a sequence of functions where $f_i: \{0, 1\}^i \rightarrow \{0, 1\}$ describes the (deterministic) strategy of agent $i$, mapping all the information she sees to the action she takes:
  $a_i = f_i(a_1, \ldots, a_{i-1}, s_i)$.
Note that when we fix a deterministic protocol, the probability $p_i$ of agent $i$ answering correctly is only over the randomness in the signals $s_1, \ldots, s_i$ drawn from $D_\theta$.

%!TEX root = paper.tex

\section{Our Deterministic Protocol}

We start with two characterization results as a warm up before stating our protocol.

Fix a deterministic protocol $\Pi$.

{\bf Fact 1.} Suppose that for some $\eps > 0$, infinitely many agents disregard the information provided by their predecessors and act only based on their private signals with probability at least $\eps$. Then learning does not occur under $\Pi$.

{\bf Fact 2.} Suppose that there is a nonzero probability that only finitely many agents take an action which depends on their private signal (and not just their predecessors' actions). Then asymptotic learning does not occur under $\Pi$. 

The intuition behind Fact 1 is that whenever an agents relies on her own signal, she has a constant probability of choosing an incorrect action. Fact 2 is a consequence of the observation that if only finitely many agents rely on their own signals then an outside observer cannot learn the state with probability tending to one, and hence neither can the agents.

Thus, in any protocol that achieves asymptotic learning, infinitely many agents must rely on their signals, but the probability that agent $i$ relies on her signal must vanish as $i$ tends to infinity. This is easily achieved using a randomized protocol, by asking agent $i$ to act according to her signal with (diminishing) probability $1/i$. Since $1/i$ is not summable, by the Borel-Cantelli lemma there will almost surely be infinitely many agents who reveal their signal.

Our deterministic protocol follows the same intuition.
We want infinitely many agents to reveal their signals, but each individual agent to reveal with diminishing probabilities.
Our algorithm can be viewed as derandomizing the random protocol described above, using the randomness in the agents' private signals.
The actions of the revealing agents (which are equal to their private signals) are observed by other agents to learn about $\theta$, and are simultaneously used as randomness to select the next revealing agent.
It is worth noting that even though the sequence of revealing agents is random, each agent (who observes the entire history of actions) knows exactly which subset of predecessors revealed their private signals.

% \iffalse
\begin{figure}[h]
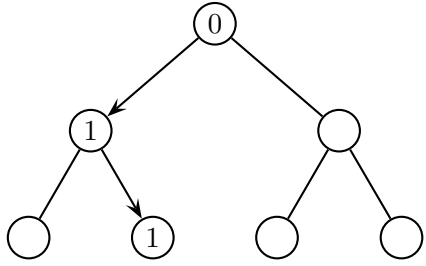

\centering
\psset{unit=0.8cm, arrowsize=0.2 1, fillstyle=solid}
$
\psmatrix[colsep=0.3, rowsep=1, mnode=circle]
            &   &   & 0                                          \\
            & 1 &   &   &             & \phantom{0}              \\
\phantom{0} &   & 1 &   & \phantom{0} &             & \phantom{0}
\ncline{->}{1,4}{2,2}
\ncline{-}{2,2}{3,1}
\ncline{->}{2,2}{3,3}
\ncline{-}{1,4}{2,6}
\ncline{-}{2,6}{3,5}
\ncline{-}{2,6}{3,7}
\endpsmatrix
$
\caption{An example of how our deterministic protocol selects the revealing agents. We assign the agents to nodes on a complete binary tree level by level, the private signal of the revealing agents are labelled on their nodes.}
\label{fig:braess}
\end{figure}
% \fi

Formally, let $t_k$ denote the index of the $k$-th revealing agent.
We define $t_1 = 1$, $t_2 = a_{t_1} + 2$, and
\[
t_k = \sum_{j=1}^{k-1} a_{t_j} \cdot 2^{j-1} + 2^{k-1}.
\]
The definition of $t_k$ partitions the agents into groups of size $2^j$ and picks exactly one agent in each group:
  $t_1 \in \{1\}, t_2 \in \{2, 3\}, \ldots, t_k \in \{2^{k-1}, \ldots, 2^{k}-1\}$.
The actions of all the revealing agents $(a_{t_{k-1}}, \ldots, a_{t_1})$ can be viewed as a number represented in binary,
  which decides the index of the next revealing agent $t_k$.
We define the strategy of each agent depending on if she is revealing or not.
For a given $i$, let $k$ be the unique integer with $2^{k-1} \leq i < 2^k$. Then
\[ f_i(a_1, \cdots, a_{i-1}, s_i) = \begin{cases} g_k(a_{t_1}, \dots, a_{t_{k-1}}, s_i) & \text{ if } i \neq t_k, \\
  s_i & \text{ if } i = t_k \end{cases} \]
  
where

\[ g_k(x_1, \dots, x_k) = \begin{cases} 0 & \text{ if } \frac{1}{k}\sum_{j=1}^{k} x_j \leq \frac{q_0 + q_1}{2}, \\
  1 & \text{ if } \frac{1}{k}\sum_{j=1}^{k} x_j > \frac{q_0 + q_1}{2}. \end{cases} \]

Theorem~\ref{thm:main} is a direct corollary of Lemma~\ref{lem:main}.
\begin{lemma}
\label{lem:main}
Let $q_0, q_1$ denote the parameters of the Bernoulli distributions $D_0$ and $D_1$ respectively.
We assume that $\eps < q_0 < q_1 < 1-\eps$ and $q_1 - q_0 > 2\eps$ for some constant $\eps > 0$.
%\footnote{In other words, the distribution $D_0$ and $D_1$ is fixed, while the number of agents grows to infinity.}
The protocol $\{f_i\}_{i \in \NN}$ defined above satisfies $p_n \ge 1 - 2n^{-\eps^2}$ for all $n \in \NN$.
\end{lemma}
Lemma~\ref{lem:main} follows immediately from  Lemmas~\ref{lem:prob-on-path}~and~\ref{lem:prob-chernoff}.
% We defer the proof of Lemma~\ref{lem:main} to the end of this section.

\begin{lemma}
\label{lem:prob-on-path}
The probability of agent $n$ being a revealing agent is at most $n^{-\eps}$.
\end{lemma}
\begin{proof}
Fix an index $i$ with $2^{k-1} \le i < 2^k$, i.e., agent $i$ is on the $k$th level of the binary tree.
If agent $i$ is selected as a revealing agent, we must have $t_k = i$.
This happens if and only if the private signals of the first $k-1$ revealing agents form the binary representation of the number $(i-2^{k-1})$.
In other words, there is a unique path from the root of the tree to agent $i$ and we must always take the correct edges to reach $i$.
Observe that the private signals of the first $k-1$ revealing agents are drawn i.i.d. from $D_\theta$.
Since $\eps < q_0 < q_1 < 1-\eps$, each private signal matches the binary representation of $i$ with probability at most $1-\eps$.
Therefore, we conclude that the probability agent $n$ with $n \ge 2^{k-1}$ is on the paths is at most
\[
(1-\eps)^{k-1} \le e^{-\eps(k-1)} = \left(2^{k-1}\right)^{-\eps \log_2(e)} \le n^{-\eps \log_2(e)} \le n^{-\eps}.
\]
where we use in the first step that $1-x \le e^{-x}$ for all $x \in \RR$.
\end{proof}

\begin{lemma}
\label{lem:prob-chernoff}
When agent $n$ is not a revealing agent, she answers correctly with probability at least $1-n^{-\eps^2}$.
\end{lemma}
\begin{proof}
Fix an index $n$ with $2^{k-1} \le n < 2^{k}$.
Agent $n$ has access to the actions of the revealing agents $t_1, \ldots, t_{k-1}$ since $t_k < 2^{k-1}$.
Because these agents reveal their private signals, agent $i$ has at least $k$ i.i.d. samples from $D_\theta$ (including her own signal).
She can simply take the average of these samples and check if the empirical mean is closer to $q_0$ or $q_1$ to guess the hidden state $\theta$.
Note that most of the agents are not revealing, and they are acting rationally.

Let $\bar q = \frac{q_0 + q_1}{2}$.
Since we assumed that $q_1 - q_0 \ge 2\eps$, we have that $|\bar q - q_\theta| \ge \eps$.
Therefore, as long as the empirical mean of the $k$ samples has additive error less than $\eps$, agent $i$ will be able to guess $\theta$ correctly.
By standard application of Chernoff-Hoeffding bounds, we have
\[
\Pr\left[ \left| \frac{1}{k}\left(\sum_{j=1}^{k-1} a_{t_j} + s_i\right) - q_\theta \right| \ge \eps \right] \le \exp(-2k\eps^2) \le n^{-\eps^2}. \qedhere
\]
\end{proof}

We are now ready to prove Lemma~\ref{lem:main}.

\begin{proof}[Proof of Lemma~\ref{lem:main}]
We prove the protocol $\{f_i\}_{i \in \NN}$ defined as above satisfies $p_n \ge 1-{2n^{-\eps^2}}$.
The claim then follows from the assumption that $\eps > 0$ is a constant.

Agent $n$ can pick the wrong action due to one of the two reasons:
\begin{enumerate}
\item Agent $n$ is revealing, and her private signal is different from $\theta$. The probability of this event is upper bounded by the probability that agent $n$ is revealing, which is $n^{-\eps}$ by Lemma~\ref{lem:prob-on-path}. 
\item The agent is not revealing, but the previous information leads to the wrong conclusion on $\theta$. By Lemma~\ref{lem:prob-chernoff}, this happens with probability at most  $n^{-\eps^2}$.
\end{enumerate}
We conclude the proof by taking a union bound over these two cases.
\end{proof}

%\section*{Acknowledgment}

\bibliographystyle{plain}
\bibliography{herding}

\begin{thebibliography}{1}

\bibitem{acemoglu2011bayesian}
Daron Acemoglu, Munther~A Dahleh, Ilan Lobel, and Asuman Ozdaglar.
\newblock Bayesian learning in social networks.
\newblock {\em The Review of Economic Studies}, 78(4):1201--1236, 2011.

\bibitem{banerjee1992simple}
Abhijit~V Banerjee.
\newblock A simple model of herd behavior.
\newblock {\em The Quarterly Journal of Economics}, pages 797--817, 1992.

\bibitem{bikhchandani1992theory}
Sushil Bikhchandani, David Hirshleifer, and Ivo Welch.
\newblock A theory of fads, fashion, custom, and cultural change as
  informational cascades.
\newblock {\em Journal of political Economy}, pages 992--1026, 1992.

\bibitem{chamley2004rational}
Christophe Chamley.
\newblock {\em Rational herds: Economic models of social learning}.
\newblock Cambridge University Press, 2004.

\bibitem{drakopoulos2013learning}
Kimon Drakopoulos, Asuman Ozdaglar, and John~N Tsitsiklis.
\newblock On learning with finite memory.
\newblock {\em IEEE Transactions on Information Theory}, 59(10):6859--6872,
  2013.

\bibitem{eyster2010naive}
Erik Eyster and Matthew Rabin.
\newblock Naive herding in rich-information settings.
\newblock {\em American economic journal: microeconomics}, 2(4):221--243, 2010.

\bibitem{hann2017speed}
Wade Hann-Caruthers, Vadim~V Martynov, and Omer Tamuz.
\newblock The speed of sequential asymptotic learning.
\newblock 2017.

\bibitem{lobel2009rate}
Ilan Lobel, Daron Acemoglu, Munther Dahleh, and Asuman Ozdaglar.
\newblock Rate of convergence of learning in social networks.
\newblock In {\em Proceedings of the American Control Conference}, 2009.

\bibitem{smith2000pathological}
Lones Smith and Peter S{\o}rensen.
\newblock Pathological outcomes of observational learning.
\newblock {\em Econometrica}, 68(2):371--398, 2000.

\end{thebibliography}

\end{document}